\documentclass[a4paper,UKenglish,cleveref, autoref, thm-restate]{lipics-v2021}
\hideLIPIcs
\nolinenumbers
\pdfoutput=1
\usepackage{algorithm2e}

\usepackage{xspace}
\usepackage{amsmath}

\newcommand\problem[4]{
\begin{center}
    \noindent
    \framebox{
    \begin{minipage}{0.94\textwidth}
        {\sc #1}\par
        {\bf Input:} #2\par
        {\bf Parameter:} #3\par
        {\bf Output:} #4
    \end{minipage}}
\end{center}
}

\newcommand{\FPT}{{\sc FPT}\xspace}
\newcommand{\bigO}{\mathcal{O}}

\title{On the parameterized complexity of symmetric directed multicut}

\author{Eduard Eiben}{Department of Computer Science, Royal Holloway, University of London, Egham, UK}{Eduard.Eiben@rhul.ac.uk}{https://orcid.org/0000-0003-2628-3435}{}
\author{Cl\'ement Rambaud}{DIENS, École Normale Supérieure, CNRS, PSL University, Paris, France}{clement.rambaud@ens.psl.eu}{}{}
\author{Magnus Wahlstr\"om}{Department of Computer Science, Royal Holloway, University of London, Egham, UK}{Magnus.Wahlstrom@rhul.ac.uk}{https://orcid.org/0000-0002-0933-4504}{}
\authorrunning{E. Eiben, C. Rambaud, and M. Wahlstr\"om}
\Copyright{E. Eiben, C. Rambaud, and M. Wahlstr\"om}

\begin{document}

\maketitle

\begin{abstract}
    We study the problem \textsc{Symmetric Directed Multicut} from a parameterized complexity perspective. 
    In this problem, the input is a digraph $D$, a set of \emph{cut requests} $C=\{(s_1,t_1),\ldots,(s_\ell,t_\ell)\}$
    and an integer $k$, and the task is to find a set $X \subseteq V(D)$ of size at most $k$ such that for every $1 \leq i \leq \ell$,
    $X$ intersects either all $(s_i,t_i)$-paths or all $(t_i,s_i)$-paths. Equivalently, every strongly connected component
    of $D-X$ contains at most one vertex out of $s_i$ and $t_i$ for every $i$. This problem is previously
    known from research in approximation algorithms, where it is known to have an $O(\log k \log \log k)$-approximation.
    We note that the problem, parameterized by $k$, directly generalizes multiple interesting FPT problems
    such as \textsc{(Undirected) Vertex Multicut} and \textsc{Directed Subset Feedback Vertex Set}. 
    We are not able to settle the existence of an FPT algorithm parameterized purely by $k$,
    but we give three partial results: An FPT algorithm parameterized by $k+\ell$; an FPT-time
    2-approximation parameterized by $k$; and an FPT algorithm parameterized by $k$ for the special
    case that the cut requests form a clique, \textsc{Symmetric Directed Multiway Cut}. 
    The existence of an FPT algorithm parameterized purely by $k$ remains an intriguing open possibility.
\end{abstract}

\section{Introduction}

Graph separation problems have been studied in parameterized complexity 
for a long time, and with significant success. In particular for undirected
graphs, a wide range of powerful FPT algorithms have been constructed,
from the early results on \textsc{Odd Cycle Transversal} by Reed et al.~\cite{ReedSV04}
and \textsc{Multiway Cut} by Marx~\cite{Marx2006},
to quite generic problems such as \textsc{Vertex Multicut}~\cite{BousquetDT18SICOMP,marx2013fixedparameter}.
In the latter problem, the input is an undirected graph $G$, 
a set of \emph{cut requests} $C=\{(s_1,t_1), \ldots, (s_\ell,t_\ell)\}$, and an integer $k$,
and the goal is to find, if it exists, a set of at most
$k$ vertices whose removal disconnects $s_i$ from $t_i$, for every $1 \leq i \leq \ell$.
Marx showed an FPT algorithm for this problem parameterized by $k+\ell$~\cite{Marx2006},
but the question of an FPT algorithm parameterized by $k$ alone remained open
for a long time, until finally settled simultaneously by Bousquet et al.~\cite{BousquetDT18SICOMP}
and Marx and Razgon~\cite{MarxR14SICOMP}.

For directed graphs, by comparison, the success is more limited, and
the line between FPT and W[1]-hard cut problems is much less clear. On
the one hand, some high profile FPT algorithms do exist for directed
graph problems. One of the earliest was \textsc{Directed Feedback
  Vertex Set}, where the goal is to find a set of at most $k$ vertices
in a directed graph which intersects all directed cycles. This problem
was shown to be FPT in 2007 by Chen et al.~\cite{ChenLLOR08} by
reduction to an auxiliary directed graph separation problem later dubbed
\textsc{Skew Multicut}. Later FPT results, following the FPT
algorithms for \textsc{Multicut} on undirected graphs, include the problems
\textsc{Directed Multiway Cut}~\cite{ChitnisHM13dmc} and \textsc{Directed Subset Feedback
  Vertex Set}~\cite{ChitnisCHM15dsfvs}.
However, other problems which are FPT on undirected graphs are intractable on digraphs.
\textsc{Directed Odd Cycle Transversal} was shown to be W[1]-hard by Lokshtanov et al.~\cite{LokshtanovR0Z20doct},
although it admits an FPT 2-approximation. 
For another example, \textsc{Directed Multicut} is the natural generalization
of \textsc{Multicut} to digraphs. Here, the input is a digraph $D$, a set of 
cut requests $C=\{(s_1,t_1), \ldots, (s_\ell,t_\ell)\}$ and an integer $k$,
and the goal is to find, if it exists, a set of at most $k$ vertices
whose removal cuts every path from $s_i$ to $t_i$, for every $1 \leq i \leq \ell$.
This problem is W[1]-hard parameterized by $k$ alone~\cite{marx2013fixedparameter}, 
even on directed acyclic graphs (DAGs)~\cite{KratschPPW15mcdags}
or for just four cut requests~\cite{pilipczuk2018directed}.


With this background, it may be considered highly unlikely to find
a natural cut problem on digraphs that directly generalizes \textsc{Vertex Multicut}
and which is FPT parameterized by the solution size alone.
Yet, we consider a problem for which this appears intriguingly plausible.

For a first attempt at a modified problem definition, consider the variant where for every
cut request $(s_i,t_i)$ we require both directions 
$(s_i,t_i)$ and $(t_i,s_i)$ to be cut. However, 
this problem remains W[1]-hard; indeed, it is equivalent to the original
problem if the input graph is a DAG.
Furthermore, it captures
{\sc Directed Vertex Multicut} on general digraphs:
if $I=(D,T,k)$ is a {\sc Directed Vertex Multicut} instance,
construct $D'$ by adding a new vertex $s'_i$ and an arc $s'_i s_i$
for every $(s_i,t_i) \in T$. Then, there is no $(t_i,s'_i)$-path
in $D'$, and cutting every $(s'_i,t_i)$-paths and $(t_i,s'_i)$-paths
is equivalent to cut every $(s_i,t_i)$-path.
This shows that this first symmetric version of 
{\sc Directed Vertex Multicut} is $W[1]$-hard too, even for $\ell=4$.

However, another directed generalization
of {\sc Vertex Multicut} has still unknown parameterized complexity.

\problem
{Symmetric Directed Vertex Multicut}
{a digraph $D$, a set of pairs of vertices 
$C=\{(s_1, t_1), \dots, (s_\ell,t_\ell)\}$, and an integer $k$.}
{$k$}
{find, if there exists, a set $X$ of at most $k$ vertices whose
removal cuts, for every $i=1, \dots, \ell$, either all $(s_i,t_i)$-paths or all $(t_i,s_i)$-paths.}

As with many directed cut problems, there are simple reductions 
between the edge- and the vertex deletion variants. We focus on the 
vertex deletion variant since it is  easier to work with (cf.~\emph{shadow removal}, discussed below).

Let us make a few observations to get a feeling for the problem.
Let $I=(D,C,k)$ be an instance of \textsc{Symmetric Directed Vertex Multicut} 
(\textsc{Symmetric Multicut} for short), and note that
a set $X \subseteq V(D)$ is a solution if and only if
$s_i$ and $t_i$ are in distinct strongly connected components in $D-X$ for
every cut request $(s_i,t_i)$. This observation is important
for understanding the structure of the problem.

We also note that \textsc{Symmetric Multicut} generalizes several
of the above-mentioned landmark FPT problems. Indeed, first
consider \textsc{Vertex Multicut}.
Let $I=(G,C,k)$ be an instance of this problem. We can then produce
an instance $I'=(D,C,k)$ of \textsc{Symmetric Multicut} simply by replacing
every edge $uv \in E(G)$ by the arcs $uv$ and $vu$.
Indeed, for every set $X \subseteq V(D)$, the strong and weak
components of $D-X$ coincide. Hence $X$ is a symmetric multicut in $D$
if and only it is a vertex multicut in $G$.

Next, let $D$ be a digraph, and let $C=\binom{V(D)}{2}$ be the set containing
all pairs of vertices over $D$. Then $I=(D,C,k)$ captures \textsc{Directed Feedback Vertex Set}.
More generally, consider \textsc{Directed Subset Feedback Vertex Set}.
In this problem, the input is a digraph $D$, a set of arcs $S \subseteq E(D)$, 
and an integer $k$, and the goal is to find a set of at most $k$ vertices
which intersects every cycle containing an arc of $S$. By the above observation,
$I=(D,S,k)$ can be interpreted as-is as an equivalent instance of \textsc{Symmetric Multicut}.
Thus, if \textsc{Symmetric Multicut} is indeed FPT parameterized by $k$,
it would make a significant generalization over the previous state of the art.

\subparagraph{Our results}
We are not able to settle the status of \textsc{Symmetric Multicut} parameterized by $k$, 
but we give three partial results. First, we give an FPT algorithm for the combined parameter of $k+\ell$.
Second, we show an FPT 2-approximation for \textsc{Symmetric Multicut} with parameter $k$.
Finally, we consider the problem \textsc{Symmetric Directed Multiway Cut}, 
where the cut requests are a set $C=\binom{T}{2}$ containing all pairs over a set of terminals $T$;
i.e., every strongly connected component of $D-X$ is allowed to contain at most one vertex of $T$.
We show that this restricted variant is FPT parameterized by $k$.

\subparagraph{Technical overview}
The first of these results is relatively straight-forward. We consider the 
solution structure of the problem, and show 
a simple FPT reduction to \textsc{Skew Multicut}.
Since \textsc{Skew Multicut} is FPT parameterized by $k$, this finishes the result.
This is analogous to the FPT algorithm for \textsc{Vertex Multicut} parameterized
by $k+\ell$ via reduction to \textsc{Multiway Cut}, noted by Marx~\cite{Marx2006}.

The FPT 2-approximation is more interesting.  First, by iterative compression we can
assume that we have a solution $Y$, say $|Y| \leq 2k+1$, and want to
determine the existence of a solution $X$ with $|X| < |Y|$ (or otherwise
prove that there is no solution of cardinality at most $k$). 
By branching on the intersection $X \cap Y$ we can assume that no
vertex of $Y$ is to be deleted. Furthermore, recall from above
that a solution $X$ to an instance $I=(D,C,k)$ is characterized by 
the strongly connected component structure of $D-X$. Hence, we may also guess 
a partition of $Y$ into strongly connected components and a topological
order on these components. After all these steps, we have an
instance $I=(D',C,k')$ and a set $Y=\{y_1,\ldots,y_r\} \subseteq V(D)$,
such that $Y$ is a symmetric multicut for $(D,C)$ and with the assumption that
we are looking for a symmetric multicut $X$ such that $X \cap Y=\emptyset$
and in $D'-X$, $y_i$ reaches $y_j$ only if $i \leq j$.
Thus, there are two remaining tasks to coordinate. $X$
cuts all paths from $y_j$ to $y_i$ for $i<j$, and simultaneously, for every terminal $y_i$
and cut pair $(s_j,t_j)$, $X$ cuts at least one of $s_j$ and $t_j$ from the 
strongly connected component of $y_i$. We achieve a 2-approximation by treating 
these steps separately. 
The first property can be ensured by a reduction to \textsc{Skew Multicut};
we note that \textsc{Skew Multicut} is still FPT (using the 
algorithm of Chen et al.~\cite{ChenLLOR08}) even if the underlying
graph is not a DAG. The key observation is now that after deleting
such a skew multicut for $Y$, the remaining task separates into
$|Y|$ disjoint instances, one for each terminal $y \in Y$.
Hence, it remains to solve the problem for an instance where there is
a central vertex $y$ such that for every cut request $(s_i,t_i)$,
every closed walk on $s_i$ and $t_i$ passes through $y$. 
Solving this problem in FPT time finally yields and FPT-time 2-approximation
for \textsc{Symmetric Multicut}. 

The FPT algorithm for \textsc{Symmetric Directed Multiway Cut} is more technical. 
It works by adapting the algorithm for \textsc{Directed Subset Feedback Vertex Set} 
of Chitnis et al.~\cite{ChitnisCHM15dsfvs}, but there are some technical complications.
First, as a more robust formulation we consider the following setting. 
The input is a digraph $D$, a list $A_1, \ldots, A_\ell$ of sets of arcs of $D$, and an integer $k$,
with the restriction that each $A_i$ is a ``near-biclique'', $A_i=S_i \times T_i$ for some
possibly overlapping vertex sets $S_i$ and $T_i$. The task is to find a set $X \subseteq V(D)$
of at most $k$ vertices such that no closed walk in $D-X$ contains arcs from two distinct
sets $A_i$ and $A_j$. Note that this version allows us to capture both the
setting where terminals are deletable and where terminals are non-deletable,
e.g., by replacing a non-deletable terminal by $k+1$ false twins,
and for each terminal $t_i \in T$ letting $S_i$ contain the twin copies of $t_i$
and $T_i$ their out-neighbours.
More importantly, arc sets of the form $A_i=S_i \times T_i$ 
are closed under the vertex bypassing operation used in \emph{shadow removal},
which the original problem formulation is not. (See Section~\ref{sec:multiway}.)

By the same setup as the FPT 2-approximation (and as Chitnis et al.~\cite{ChitnisCHM15dsfvs}),
we reduce to the iterative compression version where we additionally
have a solution set $Y$ and an ordering $y_1<\ldots<y_r$ over $Y$,
with the assumption that $y_i$ reaches $y_j$ in $D-X$ if and only if $i<j$. 
We can now apply the shadow removal technique and consider the 
set of vertices $R$ reachable from $y_r$ in $D-X$. By shadow removal,
this set is strongly connected to $y_r$ in $D-X$. But here is the
second complication. In \textsc{Directed Subset Feedback Vertex Set}, 
$R$ cannot contain any ``terminal arc'' at all, which allows the
algorithm to proceed via an intricate branching step over graph separations
in an auxiliary graph (using the so-called \emph{anti-isolation lemma} 
and important separators branching). In our setting there can be an index $i_0$ 
such that $R$ contains arcs of $i_0$ (and $A_{i_0}$ can be unboundedly big).
However, via an extra color-coding step, we are able to modify the method
of Chitnis et al.~\cite{ChitnisCHM15dsfvs}, to allow us to guess $i_0$
and find $R$. We can then find a solution by repeating the process.
In total, we show that \textsc{Symmetric Directed Multiway Cut} has 
an algorithm in time $\bigO^*(2^{\bigO(k^3)})$.

\subparagraph{Related work}
The problem \textsc{Symmetric Multicut} was first studied by Klein et al.~\cite{KleinPRT97jalg} 
in the context of approximation algorithms. The results were improved upon
by Even et al.~\cite{EvenNRS00jacm}, who showed that \textsc{Symmetric Multicut}
admits an $O(\log k \log \log k)$-approximation, where $k$ is the size of the optimal 
solution. By contrast, the best approximation ratio we are aware of for \textsc{Directed Multicut}
is just slightly better than $O(\sqrt{n})$ (Agarwal et al.~\cite{AgarwalAC07cuts}, improving on 
previous work~\cite{CheriyanKR05dmc,Gupta03}). Chuzhoy and Khanna~\cite{ChuzhoyK09cuts} showed 
that achieving a subpolynomial approximation ratio for \textsc{Directed Multicut} is hard.

We will make use of much of the toolbox developed for FPT algorithms for 
graph separation problems. In particular, the method of \emph{iterative compression},
first used for \textsc{Odd Cycle Transversal} by Reed et al.~\cite{ReedSV04};
the notion of \emph{important separators}, which underpins Marx' results
on \textsc{Multiway Cut} and related problems~\cite{Marx2006};
and the notion of \emph{shadow removal}, developed by Marx and Razgon
for \textsc{Vertex Multicut}~\cite{marx2013fixedparameter}.
These notions are explained in Section~\ref{sec:preliminaries}.
The work that is closest to our results is the FPT algorithm
for \textsc{Directed Subset Feedback Vertex Set} of Chitnis et al.~\cite{ChitnisCHM15dsfvs}.

Kim et al.~\cite{KimKPW22stoc} recently further extended the toolbox for directed
graph separation problems by a method of \emph{flow augmentation} for directed graph cuts.
This settled several long-standing problems, among other results developing an FPT
algorithm for the notorious \textsc{$\ell$-Chain SAT} problem.
Unfortunately, this method is not directly applicable to \textsc{Symmetric Multicut}
as the cut structure in the latter problem is more complex than simple $(s,t)$-cuts.

Ramanujan and Saurabh~\cite{RamanujanS17skew} considered \textsc{Skew-Symmetric Multicuts},
a problem family of multicuts on \emph{skew-symmetric digraphs} (which is effectively a generalization of
\textsc{Almost 2-SAT}). However, except for the problem name, this bears
no relation to \textsc{Symmetric Multicut}, as studied in this paper,
or to \textsc{Skew Multicut}, the auxiliary problem in the classic FPT algorithm
for \textsc{Directed Feedback Vertex Set}~\cite{ChenLLOR08}.

\subparagraph{Structure of the paper}
After introducing some useful tools in Section~\ref{sec:preliminaries},
we show in Section~\ref{sec:fpt_k_l} that
{\sc Symmetric Directed Vertex Multicut} is \FPT when parameterized
by both $k$ and $\ell$.
Then, in Section~\ref{sec:2approx}, we give a $2$-approximation
algorithm with running time $f(k) n^{\bigO(1)}$.
Finally, in Section~\ref{sec:multiway}, we show that a particular
case, called {\sc Symmetric Directed Multiway Vertex Cut}, is
\FPT.

\section{Preliminaries\label{sec:preliminaries}}

\subsection{Important cuts}

In a digraph $D$, if $X,Y$ are disjoint sets of vertices,
an $(X,Y)$-cut $S$ is a set of \emph{vertices} in 
$V(D) \setminus (X \cup Y)$ such that there is no
$(X,Y)$-path in $D-S$.
A classical tool in the design of \FPT algorithms for problems
of cut in a graph is the notion of important cut.
An $(X,Y)$-cut is said to be important if there is no
$(X,Y)$-cut further from $X$ with smaller or equal size.

\begin{definition}
    Let $D$ be a digraph and $X,Y$ be two disjoint sets of vertices.
    An $(X,Y)$-cut $S$ with set $R$ of vertices reachable from
    $X$ in $D-S$ is said to be \emph{important} if
    \begin{enumerate}
        \item $S$ is an inclusion-wise minimal $(X,Y)$-cut, and
        \item there is no $(X,Y)$-cut $S' \neq S$ of size at most $|S|$
              such that the set of vertices reachable from $X$
              in $D-S'$ is a superset of $R$.
    \end{enumerate}
    Symmetrically, $S$ is said to be \emph{anti-important}
    if it is an important $(Y,X)$-cut in $D^{op}$, the digraph obtained
    from $D$ by reversing every arc.
\end{definition}

All fundamental results on important cuts are summarised
in the following property. We refer the reader to 
\cite[Part 8.5]{cygan2015parameterized} for proofs.

\begin{proposition}
Let $D$ be a digraph, $X,Y$ be disjoint sets of vertices
and $k$ be an integer.
\begin{enumerate}
    \item One can test in polynomial time whether an $(X,Y)$-cut
          $S$ is important.
    \item\label{prop:ext_important_cut}
          If $S$ is an $(X,Y)$-cut with set $R$ of vertices reachable 
          from $X$ in $D-S$, one can compute in
          polynomial time an important $(X,Y)$-cut $S'$ such that
          $|S'| \leq |S|$ and the set of vertices reachable
          from $X$ in $D-S'$ contains $R$.
    \item If $\mathcal{S}$ is the set of important $(X,Y)$-cuts,
          then $\sum_{S \in \mathcal{S}} 4^{-|S|} \leq 1$.
    \item\label{prop:enum_important_cut}
          If $\mathcal{S}_k$ is the set of important $(X,Y)$-cuts
          of size at most $k$, then $|\mathcal{S}_k| \leq 4^k$
          and $\mathcal{S}_k$ can be enumerated in time 
          $4^k n^{\bigO(1)}$.
\end{enumerate}
\end{proposition}

\subsection{Iterative compression}
Iterative compression is a standard method
in the design of \FPT algorithms.

To avoid repetition, we give here a general property to
deal with iterative compression.
Let $\mathcal{L}$ be a parameterized algorithmic problem
such that an instance of $\mathcal{L}$ has the form
$I = (D,T,k)$ where $D$ is a digraph, $T$ depends on the problem
and $k$ is an integer.
We suppose a few properties on $\mathcal{L}$:
\begin{itemize}
    \item an instance $I=(D,T,k)$ is a yes-instance if and only if there exists
      a set $X$ of at most $k$ vertices satisfying a given property $P(D,T,X)$,
      which is supposed to be checkable in polynomial time,
    \item if $D$ is empty, then $\emptyset$ is a solution, and
    \item for every vertex $v \in V(D)$, if $X$ satisfies $P(D-v,T,X)$,
        then $X\cup\{v\}$ satisfies $P(D,T,X\cup\{v\})$.
\end{itemize}
These three properties will clearly hold for every problems 
considered in this paper.

We say that an algorithm $\mathcal{A}$ is an $\alpha$-approximation
for some $\alpha\geq 1$ if for every input instance $(D,T,k)$,
either it concludes that there is no solution of size at most $k$,
or it returns a solution of size at most $\alpha k$.
For $\alpha=1$, this is an exact algorithm.

We now define the \emph{compression} problem $\mathcal{L}'$ by: 
given $I'=(D,T,Y,k)$ where $(D, T, Y)$ satisfies $P$,
find a solution of the $\mathcal{L}$ instance $(D,T,k)$.
The parameters are now $(k,|Y|)$.
The compression problem is equivalent to the original one in the following
sense:

\begin{proposition}\label{prop:iterative_compression}
Let $\alpha \geq 1$, and $t(k,|Y|)$ be a real function
which is increasing for each parameter if the other one is fixed,
and $c \geq 0$ a constant.
If there exists an algorithm $\mathcal{A}'$ finding
an $\alpha$-approximation for $\mathcal{L}'$ in time $t(k,|Y|) n^c$
then there exists an algorithm $\mathcal{A}$
finding an $\alpha$-approximation for $\mathcal{L}$
in time $t(k, \alpha k + 1) n^{c+1}$.
In particular, if $\mathcal{L}'$ is \FPT, then $\mathcal{L}$ is \FPT
too.
\end{proposition}

The proof is in the appendix.
For further information on iterative compression
we refer to \cite[Chapter 4]{cygan2015parameterized}.

\subsection{A general framework for shadow removal}

The concept of shadow was first introduced by Marx and Razgon~\cite{marx2013fixedparameter}.
The idea is to make the problem easier by assuming that
there exists a solution $X$ such that every vertex 
$v \in V(D) \setminus X$ is reachable from a given
set of vertices $T$, and can also reach $T$ in $D-X$.
Here, we give a general framework that was designed by
Chitnis et al.~\cite{ChitnisCHM15dsfvs}.

Let $D$ be a digraph and $T$ a set of vertices.
For every set of vertices $X$ disjoint from $T$, we define
the \emph{shadow} of $X$ to be the set of vertices
in $V(D) \setminus (T \cup  X)$ that either can not reach
$T$ in $D-X$, or are not reachable from $T$ in $D-X$.
Chitnis et al.~\cite{ChitnisCHM15dsfvs} provided a set 
of sufficient conditions under which we can comupte an 
over-approximation of the shadow of a solution to a problem;
in other words, we can compute a set $W$, disjoint from $T$, such that there exists a 
solution $X$, disjoint from $W$, where the shadow of $X$
is contained in $W$.

To state the result we need a few definitions from Chitnis et al.~\cite{ChitnisCHM15dsfvs}.  

\begin{definition}
Let $\mathcal{F} = \{F_1, \dots, F_q\}$ be a set of subgraphs
of $D$. We say that $\mathcal{F}$ is \emph{$T$-connected} if for every
$i=1, \dots,q$, every vertex in $F_i$ can reach $T$ by a walk
completely in $F_i$, and is reachable from $T$ by a walk completely
in $F_i$. A set of vertices $X \subseteq V(D)$ is said to be an
\emph{$\mathcal{F}$-transversal} if for every $i \in \{1, \dots, q\}$,
$F_i \cap X \neq \emptyset$.
\end{definition}

For example, if $\mathcal{F}$ is a set of walks, as is the case in our application,
then $X$ is an $\mathcal{F}$-transversal if and only if $X$ cuts every
walk in $\mathcal{F}$.
We can now give the main theorem that gives a superset of the shadow.

\begin{theorem}[\cite{ChitnisCHM15dsfvs}]
\label{thm:cover_shadow}
Let $T \subseteq V(D)$ and $k \in \mathbb{N}$.
One can construct in time $2^{\bigO(k^2)}n^{\bigO(1)}$ a family
$Z_1, \dots, Z_t$ of $t=2^{\bigO(k^2)} \log^2 n$ sets of vertices
such that for any set $\mathcal{F}$ of $T$-connected subgraphs of $D$,
if there exists an $\mathcal{F}$-transversal of size
at most $k$, then there exists an $\mathcal{F}$-transversal
$X$ and $i \in \{1, \dots, t\}$ such that:
\begin{enumerate}
    \item $|X| \leq k$,
    \item $X \cap Z_i = \emptyset$,
    \item the shadow of $X$ is included in $Z_i$.
\end{enumerate}
\end{theorem}

\subsection{{\sc Skew Vertex Multicut} is \FPT}
In this section, we present a problem which is known to be
\FPT. This problem was first introduced by \cite{ChenLLOR08}
in the first proof that {\sc Directed Feedback Vertex Set} is \FPT.

\problem
{Skew Vertex Multicut}
{a digraph $D$, an ordered list of pair of vertices 
$(s_1,t_1), \dots, (s_r,t_r)$
and an integer $k$.}
{$k$}
{find, if there exists, a set $X$ of at most $k$ vertices
such that there is no $(s_j,t_i)$-path in $D-X$ if
$j\geq i$.}

\begin{theorem}[\cite{ChenLLOR08}]\label{thm:skew_is_fpt}
The problem {\sc Skew Vertex Multicut} is \FPT
and can be solved in time $\bigO(4^k k^3 n^2)$.
\end{theorem}

\section{An \FPT algorithm when parameterized by \texorpdfstring{$k+\ell$}{k+l}\label{sec:fpt_k_l}}

This section aims to prove the following theorem
(remember that in {\sc Symmetric Directed Vertex Multicut}, $k$ is the size of
the desired solution, and $\ell$ is the number of cut requests).

\begin{theorem}\label{thm:fpt_with_k_l}
There is an algorithm that solves
{\sc Symmetric Directed Vertex Multicut}
in time $\bigO\left((2\ell+1)^{2\ell} 4^k k^3 n^2\right)$.
\end{theorem}

\begin{proof}
Let $I=(D,C,k)$ be a {\sc Symmetric Directed Vertex Multicut}
instance.
We suppose that $I$ is a yes-instance and let $X_{OPT}$
be a solution for $I$. Let $T=\bigcup_{(s,t) \in C} \{s,t\}$.

Let $T_0, T_1, \dots, T_r$ with $r \leq 2\ell$ be a partition
of $T$ such that:
\begin{itemize}
    \item $T_0 = X_{OPT} \cap T$,
    \item for every $i \in \{1, \dots, r\}$ and every
        $t,t' \in T_i$, $t$ and $t'$ are strongly connected
        in $D-X_{OPT}$,
    \item there is no $(T_j,T_i)$-path in $D-X_{OPT}$ if $j>i$.
\end{itemize}
Such a partition exists: consider the strongly connected components of $D-X_{OPT}$
and order them into a topological order $C_1, C_2, \dots, C_r$,
that is an ordering such that for every arc $uv$ in $D-X_{OPT}$
with $u \in C_i$ and $v \in C_j$, we have $i \leq j$.
Then set $T_i = C_i \cap T$ for every $i \in \{1, \dots, r\}$.

The first step of our algorithm guesses that partition,
thereby multiplying the running time by at most $(2\ell+1)^{2\ell}$.
Reject any partition where $s, t \in T_i$ for any $(s,t) \in C$ and any $i$.
Now, we consider the digraph $D'$ obtained by removing
$T_0$ from $D$ and merging each $T_i$ into a single vertex $t_i$,
for every $i=1, \dots, r$.

Let $I' = (D', \{(t_1,t_2), \dots, (t_{r-1},t_r)\}, k-|T_0|)$,
a {\sc Skew Vertex Multicut} instance.
Clearly, $X_{OPT} \setminus T_0$ is a solution for $I'$,
by definition of $T_0, \dots, T_r$.
Reciprocally, if $I'$ has a solution $X'$, then
consider $X = T_0 \cup X'$, which has size at most $(k-|T_0|)+|T_0| = k$.
If $X$ is not a solution for $I$, then there exists
$(s,t) \in C$ strongly connected in $D-X$.
Then, $s$ and $t$ are in the same $T_i$ for some $i$,
and thus $s$ and $t$ are strongly connected in $D-X_{OPT}$,
contradicting the fact that $X_{OPT}$ is a solution for $I$.

Thus, one can solve {\sc Symmetric Directed Vertex Multicut}
by first guessing $T_0, \dots, T_r$ and then solving that
{\sc Skew Vertex Multicut} instance using Theorem~\ref{thm:skew_is_fpt}.
This algorithm has running time at most
$\bigO\left((2\ell+1)^{2\ell} 4^k k^3 n^2\right)$.
\end{proof}

\section{A \texorpdfstring{$2$}{2}-approximation algorithm\label{sec:2approx}}

In this part, we give an \FPT algorithm that finds a solution
of size at most $2k$ for {\sc Symmetric Directed Vertex Multicut}  
if it is known that there exists a solution of size at most $k$.

\subsection{Iterative compression and first guesses}

This section aims to prove that it is enough to find a $2$-approximation
algorithm for the following problem:

\problem
{Symmetric Directed Vertex Multicut Compression}
{A digraph $D$, a set of pair of vertices 
$C=\{(s_1, t_1), \dots, (s_\ell,t_\ell)\}$, an integer $k$,
and a solution $Y$ of the {\sc Symmetric Vertex Multicut} instance
$(D,C,k)$, of size at most $2k+1$, with an ordering $y_1, \dots, y_r$ of $Y$.}
{$(k,|Y|)$}
{Find, if there exists, a set $X$ of at most $k$ vertices disjoint from $Y$
such that:
\begin{enumerate}
    \item for every pair of terminals $(s,t) \in C$ with $s,t \not\in X$,
          $s$ and $t$ are not strongly connected in $D-X$, and
    \item there is no $(y_j,y_i)$-path in $D-X$ if $j>i$.
\end{enumerate}
}

\begin{proposition}
Let $t(k,|Y|)$ be a positive function that is non decreasing if one
parameter is fixed, and $c \geq 2$ a constant.

If {\sc Symmetric Directed Vertex Multicut Compression} has
a $2$-approximation algorithm $\mathcal{A}'$ with time complexity
$t(k,|Y|) n^c$,
then {\sc Symmetric Directed Vertex Multicut}
has a $2$-approximation algorithm $\mathcal{A}$ with time complexity
at most $(2k+2)^{2k+1} t(k, 2k+1) n^{c+1}$.
\end{proposition}

\begin{proof}
First, we directly apply Property~\ref{prop:iterative_compression}
with $\alpha=2$ and thus it is enough to reduce the compression 
problem of {\sc Symmetric Directed Vertex Multicut}
to {\sc Symmetric Directed Vertex Multicut Compression}.

Consider an instance $I=(D,C,k,Y)$ of
that compression problem which is supposed to be a yes-instance,
with an optimal solution $X_{OPT}$.
It is enough to show that a $2$-approximation for $I$
can be found with at most $(|Y|+1)^{|Y|}$ calls to $\mathcal{A}'$.
To do that, we guess the structure of $Y$ in $D-X_{OPT}$.
More precisely, we guess a partition of 
$Y$ into $Y_0, Y_1, \dots Y_r$
such that:
\begin{enumerate}
    \item $Y_0 = X_{OPT} \cap Y$, and
    \item if $y,y' \in Y_i$ then $y$ and $y'$ are strongly
          connected in $D-X_{OPT}$, and
    \item there is no $(Y_j,Y_i)$-path in $D-X_{OPT}$ if $j>i$.
\end{enumerate}
Such a partition exists by taking the intersection of the strongly
connected components of $D-X_{OPT}$ with $Y$.
This guess multiplies the running time by at most
$(|Y|+1)^{|Y|} \leq (2k+2)^{2k+1}$.

We now claim that the instance of the compression problem
$I'$ obtained by
\begin{enumerate}
    \item removing $Y_0$ from $D$ and decreasing $k$
        by $|Y_0|$, and
    \item merging each $Y_i$ into a single vertex $y_i$,
\end{enumerate}
is equivalent to $I$.
More precisely, if $I$ is a yes-instance, then $I'$ too
by taking $X_{OPT} \setminus Y$ as a solution. Reciprocally,
if $I'$ has a solution $X'$ of size at most $2(k-|Y_0|)$
then $X' \cup Y_0$ is a solution for $I$
of size at most $2(k - |Y_0|) + |Y_0| \leq 2k$.
This proves the property.
\end{proof}

The remaining of this section shows that 
{\sc Symmetric Directed Vertex Multicut Compression}
has a $2$-approximation algorithm.

\subsection{Finding a skew multicut of \texorpdfstring{$Y$}{Y}
\label{subsec:find_skew_multicut}}
The first step of our algorithm computes a set 
$X_0 \subseteq V(D) \setminus Y$ of at most $k$ vertices
such that there is no $(y_j,y_i)$-path in $D-X_0$ if $j>i$.

To do that, we use the problem {\sc Skew Vertex Multicut}
that is known to be \FPT.
We directly apply Theorem~\ref{thm:skew_is_fpt} to the instance
$(D, ((y_1,y_2), (y_2, y_3), \dots, (y_{r-1},y_r)), k)$
to compute a set $X_0$ of at most $k$ vertices as wanted.
Indeed, by definition of {\sc Skew Vertex Multicut}, for every $j>i$,
there is no $(y_j,y_i)$-path in $D-X_0$.
This strong property will allow us to find in the next subsection
a solution of size at most $k$ in $D-X_0$.

\subsection{Finding a solution in the simplified instance\label{subsec:find_sol_in_simplified_instance}}

This section shows how to compute a solution for
$I = (D-X_0,C,k,Y)$. This will result in a set $X_1$
of size at most $k$ such that there is no pair $s_i,t_i$
strongly connected in $D - X_0 - X_1$, that is, $X_0 \cup X_1$
is a solution of size at most $2k$.

To do that, first note that any vertex $v \in V(D)\setminus Y$
can be strongly connected with at most one vertex in $Y$ in $D-X_0$.
Our first claim shows that we can assume that exactly one vertex 
in $Y$ is strongly connected with $v$.
\begin{claim}
If $v \in V(D) \setminus (X_0 \cup Y)$ is strongly connected
to no vertex in $Y$ in $D-X_0$, then 
$I' = (D-X_0-v, C \setminus \{ab \in C \mid a=v \text{ or } b=v\}, k, Y)$
and $I$ have the same set of solutions.
\end{claim}

\begin{proof}
Clearly, if $I$ has a solution $X'$, then $X'$ is a solution for $I'$
as every closed walk in $D-X_0-v$ is also in $D-X_0$.
Reciprocally, if $X'$ is a solution for $I'$, then
adding $v$ to $D-X_0-v-X'$ does not create any
closed walk passing through at least one vertex in $Y$.
But any closed walk passing through a cut request $(s,t) \in C$
must pass through at least one vertex in $Y$. It follows
that no pair of terminals is strongly connected in $D-X_0-X'$
and $X'$ is a solution for $I$.
\end{proof}

Thus, we can remove every vertex strongly connected to no
vertex in $Y$.
We now denote by $\ell(v)$ the unique integer such that $v$ is strongly
connected with $y_{\ell(v)}$.

\begin{claim}
Let $(s,t) \in C$ be a terminal arc.
If $\ell(s) \neq \ell(t)$, then $I''=(D, C \setminus \{(s,t)\}, k,Y)$
and $I'$ have the same set of solutions.
\end{claim}

\begin{proof}
Clearly, if $I'$ has a solution, then $I''$ too.
Reciprocally, if $I''$ has a solution $X''$, then
every terminal arc different from $s,t$ is not strongly connected
in $D-X_0-X''$. But $s$ and $t$ can not be strongly connected
as $s$ and $t$ are not strongly connected in $D-X_0$.
Thus, $X''$ is a solution for $I''$ too.
\end{proof}

We now assume that for every pair of terminal $s,t$, $\ell(s)=\ell(t)$.
The next claim shows that we can process each strongly connected
component in $D-X_0$ independently.
\begin{claim}
If there is an arc $uv$ with $u$ and $v$ not strongly connected
in $D-X_0$, then $I'' = (D-uv, C, k, Y)$ and $I'$ have the same set
of solutions.
\end{claim}

\begin{proof}
If $I'$ has a solution $X'$, then $X'$ is clearly a solution for
$I''$. Reciprocally, if $X''$ is a solution for $I''$,
then adding $uv$ to $D-X_0-uv$ does not create any closed walk,
and thus $X''$ is a solution for $I'$ too.
\end{proof}

Now, we assume that $D-X_0$ has $|Y|$ weakly connected components
$Y_1, \dots, Y_r$ such that for every $i$, $V(Y_i) \cap Y = \{y_i\}$.
Observe that now the weakly connected components are strongly connected.
Let $X_{OPT}$ be an optimal solution for $I'$.
Then we guess the values $k_i = |X_{OPT} \cap Y_i|$,
which multiplies the complexity of our algorithm by at most
$(k+1)^{|Y|} = k^{\bigO(k)}$.
Now, we solve each instance $I_i = (Y_i, C, k_i, \{y_i\})$
independently.

The key result is the following ``pushing'' claim,
that shows how to construct $X_1$ as a union of important cuts.
We denote by $X_{i,OPT} = X_{OPT} \cap Y_i$ a solution of $I_i$,
that we suppose to exist.
\begin{claim}\label{claim:pushing_approx}
Let $(s,t) \in C$ be a terminal arc strongly connected in $Y_i$.
Let $(a,b) \in \linebreak \{(s,y_i),(y_i,s),(t,y_i),(y_i,t)\}$ be such that
$X_{i,OPT}$ includes an $(a,b)$-cut.
\begin{itemize}
    \item if $a=y_i$, let $S$ be the set of vertices in $X_{i, OPT}$
          with an in-neighbour reachable from $y_i$ in $Y_i-X_{i,OPT}$
          and $S'$ be the anti-important $(a,b)$-cut given by 
          Property~\ref{prop:ext_important_cut}.
          Then $X'=(X_{OPT} \setminus S) \cup S'$ is a solution
          for $I_i$ too,
    \item symmetrically,
          if $b=y_i$, let $S$ be the set of vertices in $X_{i,OPT}$
          with an out-neighbour that reaches $y_i$ in $Y_i-X_{i,OPT}$
          and $S'$ be the important $(a,b)$-cut given by 
          Property~\ref{prop:ext_important_cut}.
          Then $X'=(X_{OPT} \setminus S) \cup S'$ is a solution
          for $I_i$ too.
\end{itemize}
\end{claim}

\begin{proof}
As $s$ and $t$ are not strongly connected in $Y_i-X_{i,OPT}$,
$X_{i, OPT}$ must contain an $(a,b)$-cut for at least one
$(a,b) \in \{(s,y_i),(y_i,s),(t,y_i),(y_i,t)\}$.
It is enough to show the first point, as the second one is
the first one applied to $D^{op}$ the digraph obtained from $D$
by reversing every arc.

First, as $|S'| \leq |S|$, we have $|X'| \leq |X_{i, OPT}| \leq k_i$.
It remains to show that there is no pair $(s',t') \in C$
strongly connected in $Y_i - X'$.
Suppose that such a counterexample $(s',t')$ exists.
Then there exists a closed walk $P$ passing through $y_i$,
$s'$ and $t'$. This walk must pass through $S' \setminus S$
as it does not exist in $Y_i - X_{i,OPT}$.
But then there exists $v \in S \setminus S'$ reachable from
$y_i$ in $Y_i - S'$, contradicting the fact
that the set of vertices reachable from $y_i$ in $Y_i - S$
includes the set of vertices reachable from $y_i$ in $Y_i - S'$.
\end{proof}

We can now give the algorithm that solves $I_i=(Y_i,C,k_i,\{y_i\})$ as Algorithm~\ref{alg:sec4}.

\begin{algorithm}[hbt!]
$X_i \gets \emptyset$\;
\While{there exists $(s,t) \in C \cap V(Y_i)^2$ 
       strongly connected in $Y_i-X_i$}
{
    guess a direction $(a,b) \in \{(s,y_i),(y_i,s),(t,y_i),(y_i,t)\}$\;
    \eIf{$a=y_i$}{
        guess an anti-important $(a,b)$-cut $S'$ of size at most $k_i - |X_i|$\;
    }{
        guess an important $(a,b)$-cut $S'$ of size at most $k_i-|X_i|$\;
    }
    add $S'$ to $X_i$\;
}
return $X_i$\;
\caption{Algorithm for single-terminal case $I_i=(Y_i,C,k_i,\{y_i\})$}
\label{alg:sec4}
\end{algorithm}

If the algorithm returns a value, then it is clearly
a solution.
We now show that there exists a sequence of guesses that
leads to a solution if it exists. 
More precisely, we show that the following invariant holds:
At every iteration of the loop, there is a possible value of $X_i$
such that $X_i$ can be extended to a solution for $I_i$ if it exists.
This invariant initially holds.
If the results holds at some iteration for a set $X_i$,
let $X_{i,OPT}$ be a solution that contains $X_i$,
and for the first guess take $(a,b)$ such that $X_{i,OPT}$
contains an $(a,b)$-cut $S$.
By Claim~\ref{claim:pushing_approx} there exists an important
or anti-important $(a,b)$-cut $S'$ of size at most $|S|$ such that
$(X_{i,OPT} \setminus S) \cup S'$ is still a solution.
Thus, there exists a solution that contains $S'$ and we can safely
add it to $X_i$.

To see that the algorithm works in time $8^k n^{\bigO(1)}$,
consider the recursion tree formed by recursively branching over all possible
values of a guess, for each guess made in the algorithm.
We denote by $t(k)$ the number of leaves of this recursion tree in the worst case.
We show by induction on $k$ that $t(k)4^{-k} \leq 4^k$.
If $k=0$, the result is clear. Otherwise, if we assume the result
for smaller values of $k$, then we have 
\[
t(k)4^{-k} 
\leq 4 \sum_{S \in \mathcal{S}_k} t(k-|S|)4^{-k} 
\leq \sum_{S \in \mathcal{S}_k} t(k-|S|)4^{-(k-|S|)} 
\leq \sum_{S \in \mathcal{S}_k} 4^{k-|S|} 
\leq 4^k \sum_{S \in \mathcal{S}_k} 4^{-|S|} 
\]
where $\mathcal{S}_k$ is the set of important (or anti-important)
$(a,b)$-cuts that is enumerated in the algorithm.
It follows by Property~\ref{prop:enum_important_cut}
that $t(k) \leq 8^k$.
We note that the algorithm can easily be made deterministic
by replacing each guessing step by an exhaustive branching;
we omit the details.

These two steps give us a $2$-approximation algorithm.
\begin{theorem}
The exists an algorithm with running time
$k^{\bigO(k)} n^{\bigO(1)}$ such that given an instance of
{\sc Symmetric Directed Vertex Multicut} and an integer $k$,
either it concludes that there is no solution of size at most $k$, 
or it returns a solution of size at most $2k$.
\end{theorem}

\begin{proof}
Let $I=(D,C,k,Y)$ be a {\sc Symmetric Directed Vertex Multicut 
Compression} instance.
First, compute a skew multicut of $Y$ using
Section~\ref{subsec:find_skew_multicut}.
This gives a set $X_0$ of at most $k$ vertices, if $I$ has a solution.
Then we apply Section~\ref{subsec:find_sol_in_simplified_instance}
to find a set $X_1$ of at most $k$ vertices that is a solution
for $(D-X_0,C,k,Y)$. We can now conclude that $X_0 \cup X_1$ is
a $2$-approximation as $|X_0 \cup X_1| \leq 2k$.
\end{proof}

\section{An exact algorithm for {\sc Symmetric Directed Multiway Cut}}\label{sec:multiway}

In this section, we give an exact (i.e., non-approximate) \FPT algorithm for a particular case
of {\sc Symmetric Directed Vertex Multicut}.

\problem
{Symmetric Directed Multiway Vertex Cut}
{A digraph $D$, a set of terminals $T \subseteq V(D)$, $k \in \mathbb{N}$.}
{$k$}
{find, if there exists, $X \subseteq V(D)$ with $|X| \leq k$ 
such there is no pair of distinct terminals $t,t' \in T \setminus X$
strongly connected in $D-X$.}

\begin{theorem}\label{thm:multiway_is_fpt}
{\sc Symmetric Directed Multiway Vertex Cut} can be solved in
time $2^{\bigO(k^3)}n^{\bigO(1)}$.
\end{theorem}

Actually, we will prove that a more general problem very closely related
to {\sc Directed Subset Feedback Arc Set} is \FPT.
Chitnis et al.~\cite{ChitnisCHM15dsfvs} proved that the 
problem {\sc Directed Subset Feedback Arc Set} is \FPT.
We adapt here their method to the following problem.

\problem{Arc Terminal Symmetric Multiway Cut}
{
A digraph $D$ having possibly loops, a list $A_1, \dots, A_\ell$ of arcs in $D$,
such that for every $i$, $A_i = S_i \times T_i$ for some 
(not necessarily disjoint) sets $S_i$ and $T_i$ of vertices.
}
{$k$}
{find, if there exists, a set $X$ of at most $k$ vertices such that any closed walk in
$D-X$ intersects at most one $A_i$.}

Note that we allow repetition in the list $A_1, \dots A_\ell$.
In this case, if $A_i=A_j$ for some $i \neq j$, then every closed walk intersecting $A_i=A_j$
has to be cut.
We will call the arcs in $\bigcup_i A_i$ the \emph{terminal arcs}.

First we show that {\sc Symmetric Directed Multiway Vertex Cut} reduces to
{\sc Arc Terminal Symmetric Multiway Cut} in \FPT time.
Indeed, given an instance $I=(D,T=\{t_1, \dots t_\ell\},k)$ of 
{\sc Symmetric Directed Multiway Vertex Cut},
we consider the {\sc Arc Terminal Symmetric Multiway Cut} instance $I'=(D,(A_1, \dots A_\ell),k)$ where $A_i = \{t_i\} \times N^+_D(t_i)$.
Now one can easily see that $X$ is a solution for $I$ if and only if it is a solution
for $I'$.
Hence it is enough to find an \FPT algorithm for
{\sc Arc Terminal Symmetric Multiway Cut}.

\subsection{Iterative compression and first guesses}

By Property~\ref{prop:iterative_compression}, it is enough to find an \FPT
algorithm for the compression problem associated to \textsc{Arc Terminal Symmetric Multiway Cut}.
Thus suppose that a first solution $Y$ of size $k+1$ is given,
and we want to find a solution $X_{OPT}$ of size at most $k$.
First, we guess the intersection $Y \cap X_{OPT}$, and we remove it.
Now we assume that $X_{OPT}$ is disjoint from $Y$.
If two vertices $y,y' \in Y$ are strongly connected in $D-X_{OPT}$,
then we can merge them without breaking the solution $X_{OPT}$, and without
making the instance easier.
Now we can suppose that no two vertices in $Y$ are strongly connected in $D-X_{OPT}$.
Hence there is a topological ordering $y_1, \dots y_{|Y|}$ of $Y$
such that there is no $(y_j,y_i)$-path in $D-X_{OPT}$ if $j>i$.
Given this ordering, we can add the arc $y_iy_j$ for every $i<j$ without
breaking the solution $X_{OPT}$, and without making the instance easier.
To summarise, by multiplying the running time of the algorithm by at most
$(k+2)^{k+1} n^{\bigO(1)}$, it is enough to find an \FPT algorithm
for the following problem.

\problem{Arc Terminal Symmetric Multiway Cut Compression}
{
A digraph $D$ (having possibly loops), a list $A_1, \dots, A_\ell$ of arcs in $D$,
such that for every $i$, $A_i = S_i \times T_i$ for some 
(not necessarily disjoint) sets $S_i$ and $T_i$ of vertices,
and an ordered set $Y=(y_1, \dots,y_r)$ of vertices such that:
\begin{enumerate}
    \item for every $i \neq j$, no closed walk in $D-Y$ intersects both $A_i$ and $A_j$, and
    \item for every $1 \leq i<j \leq r$, $y_i y_j$ is an arc in $D$.
\end{enumerate}
}
{$k+r$}
{find, if there exists, a set $X$ of at most $k$ vertices such that 
\begin{enumerate}
    \item $X$ is disjoint from $Y$, 
    \item any closed walk in $D-X$ intersects at most one $A_i$, and
    \item there is no $(y_j,y_i)$-path in $D-X$ if $j>i$.
\end{enumerate}
}

\subsection{Shadow removal}\label{subsec:reduce_shadow}
Let $I=(D,(A_1, \dots A_\ell),k,Y)$ be an {\sc Arc Terminal Symmetric Multiway 
Cut Compression} instance.
To show that we can assume the solution to be shadowless,
let $\mathcal{F}$ be the family containing all closed walks
intersecting at least two distinct sets $A_i$, $A_j$ 
and all $(y_j,y_i)$-walks for $j>i$.
Note that $\mathcal{F}$ is $Y$-connected and that the problem
is precisely to find an $\mathcal{F}$-transversal $X$ disjoint from $Y$.
We apply Theorem~\ref{thm:cover_shadow} with $\mathcal{F}$, giving us a family
of $t=2^{\bigO(k^2)}\log^2 n$ sets disjoint from $Y$,
and we guess one of them, say $Z$, to be such that if
$I$ has a solution, then there exists a solution $X$
disjoint from $Z$ and with shadow contained in $Z$.
As we consider the shadow from $Y$, vertices in $Y$ can not be in the shadow
of a solution, so we can assume $Z$ and $Y$ disjoint by replacing $Z$ by $Z \setminus Y$.

We now define another instance $I/Z=(D',(A'_1, \dots, A'_\ell),k,Y)$ equivalent to $I$ in the following sense:
\begin{enumerate}
    \item if $I$ has a solution that is disjoint from $Z$ and with shadow contained in $Z$, then $I/Z$ has a shadowless solution, and
    \item if $I/Z$ has a solution, then $I$ does too.
\end{enumerate}
The construction is the following. If $D[Z]$ contains a closed walk $W$ 
such that at least two $A_i,A_j$ intersects $W$, reject $Z$.
Otherwise construct the following. Let a \emph{$Z$-walk} be a 
walk in $D$ with endpoints in $V(D')$ and internal vertices, if any, in $Z$.
\begin{itemize}
    \item $V(D') = V(D) \setminus Z$;
    \item $E(D')$ is the set of all arcs $uv$ such that there is a $Z$-walk 
          from $u$ to $v$ in $D$;
    \item for every $i=1, \dots,\ell$,
          $A'_i$ is the set of arcs $uv$ such that there is a $Z$-walk
          from $u$ to $v$ intersecting $A_i$.
          In particular, $A_i \cap E(D') \subseteq A'_i$ as a $Z$-walk can have no internal vertices.
\end{itemize}

First, we need to check that $I/Z$ is indeed an instance of
{\sc Arc Terminal Symmetric Multiway Cut Compression}
\begin{claim} \label{claim:biclique}
For every $i=1, \dots, \ell$, $A'_i = S'_i \times T'_i$ for some sets
$S'_i$ and $T'_i$ of vertices.
\end{claim}

\begin{proof}
It is enough to show that if $uv, u'v' \in A'_i$, then $uv' \in A'_i$.
By definition, there exists a $Z$-walk $W$ (resp. $W'$) from $u$ to $v$ 
(resp. $u'$ to $v'$), with possibly no internal vertices, which goes through a 
terminal arc $ab \in A_i$ (resp. $a'b' \in A_i$), where the terminal arc may be a loop.
As $A_i = S_i \times T_i$, we have $ab' \in A_i$, and so by combining
a prefix of $W$ with a suffix of $W'$, there is a $Z$-walk
from $u$ to $v'$ containing an arc in $A_i$. This shows that $uv' \in A'_i$.
\end{proof}

\begin{claim}
$I/Z$ is an instance of \textsc{Arc Terminal Symmetric Multiway Cut Compression}.
\end{claim}
\begin{proof}
By Claim~\ref{claim:biclique}, $A_i'=S_i' \times T_i'$ for every $i$,
and the arcs $y_iy_j$, $i<j$ remain in $D'$.
It remains to check that $Y$ is a solution for $D'$. Assume to the contrary,
and let $W$ be a closed walk in $D'-Y$ intersecting two sets $A_i$ and $A_j$, $i \neq j$. 
But then $W$ expands into a closed walk $W'$ in $D$ by replacing every arc of $W$ with
a corresponding $Z$-walk. Since $Y \cap Z = \emptyset$, this is a closed walk in $D$
intersecting $A_i$ and $A_j$, disjoint from $Y$. This is a contradiction. 
\end{proof}

\begin{claim}
If $I$ has a solution disjoint from $Z$ and with shadow contained in $Z$,
then $I/Z$ has a shadowless solution.
\end{claim}

\begin{proof}
Let $X$ be a solution of $I$ disjoint from $Z$ and with shadow contained in $Z$.
We claim that $X$ is a shadowless solution of $I/Z$.

First, let's see why $X$ is a solution of $I/Z$.
Suppose for contradiction that $D'-X$ contains a closed walk $W'$
containing two terminal arcs $uv \in A'_i$ and $u'v' \in A'_j$ for some
distinct indices $i$ and $j$.
Then we construct a closed walk $W$ in $D-X$ intersecting both $A_i$ and $A_j$:
replace in $W'$ the arc $uv$ (resp. $u'v'$) by a $Z$-walk from $u$ to $v$ 
(resp. $u'$ to $v'$) intersecting $A_i$ (resp. $A_j$), 
and for every other arc $xy \in W'$ which is not in $D$, replace $xy$
by a $Z$-walk from $x$ to $y$.
This gives a closed walk $W$ in $D-X$ intersecting both $A_i$ and $A_j$,
contradicting the fact that $X$ is a solution of $I$.
Similarly, if there is a $(y_j,y_i)$-path $P'$ in $D'-X$ for some $j>i$,
then we can expand $P'$ into a $(y_j,y_i)$-walk $W$ in $D-X$,
which can be shortcut into a $(y_j,y_i)$-path $P$ in $D-X$.

Now we show that $X$ is shadowless in $I'$.
For every vertex $u \in V(D) \setminus Z$, we know that there is a $(u,Y)$-path
$P^+$ (resp. $(Y,u)$-path $P^-$) in $D-X$, as the shadow of $X$ is included in $Z$.
Then we replace every $Z$-walk in $P^+$ (resp. $P^-$) by the arc linking its endpoints.
This gives a $(u,Y)$-path (resp. $(Y,u)$-path) in $D'-X$, and so $v$ is not in the
shadow.
This proves that $X$ is shadowless in $D'$.
\end{proof}

\begin{claim}
If $I/Z$ has a solution then $I$ too.
\end{claim}

\begin{proof}
Suppose that $I/Z$ has a solution $X$.
We claim that $X$ is a solution for $I$ too.

Suppose for contradiction that $D-X$ has a closed walk $W$ intersecting
both $A_i$ and $A_j$ for some distinct indices $i$ and $j$.
Then construct the closed walk $W'$ in $D'-X$ as follows:
replace every $Z$-walk in $W$ by the arc linking its endpoints.
This creates a closed walk $W'$ in $D'-X$ intersecting both $A'_i$ and $A'_j$,
contradicting the fact that $X$ is a solution for $I'$.
A similar step applies if $D-X$ contains a $(y_j,y_i)$-path for some $j>i$.
\end{proof}

As a consequence, we are able to transform the original instance $I$ into an
equivalent instance $I/Z$ which has a shadowless solution.
Guessing $Z$ multiplies the running time by at most $2^{\bigO(k^2)}\log^2 n$,
and then computing $I/Z$ is performed in polynomial time.

\subsection{Finding a shadowless solution}

We now suppose that $I=(D,(A_1, \dots A_\ell),k,Y)$ has a shadowless solution
$X_{OPT}$. Remember that $y_1, \dots, y_r$ is an ordering of $Y$
such that there is no $(y_j,y_i)$-path in $D-X_{OPT}$ if $j>i$,
and  for every $j>i$, $y_i y_j$ is an arc in $D$.
As the solution $X_{OPT}$ we are searching for is shadowless, 
every vertex in $D-X_{OPT}$ reaches $Y$, and so $y_r$ (because
$y_r$ is dominated by $Y \setminus \{y_r\}$).

Another observation is that for at most one index ${i_0}$, $A_{i_0}$
contains a terminal arc strongly connected with $y_r$ in $D-X_{OPT}$.
In what follows, we implicitly suppose that $i_0$ exists, otherwise
we can set by convention $A_{i_0} = \emptyset$.
As $X_{OPT}$ is shadowless, an arc $uv$ is strongly
connected with $y_r$ in $D-X_{OPT}$ if and only if
\begin{enumerate}
    \item $y_r$ reaches $u$ in $D-X_{OPT}$ and
    \item $v \not\in X_{OPT}$.
\end{enumerate}

The next claim allows us to find the set of vertices $v$ which violates
the second condition. Let $R$ denote the set of vertices reachable from 
$y_r$ in $D-X_{OPT}$ and note by shadowlessness that $R$ precisely
describes the strongly connected component of $y_r$ in $D-X_{OPT}$.
Say that $A_i$ is \emph{active in $X_{OPT}$} if $i \neq i_0$ and 
$S_i \cap R \neq \emptyset$ (and note that this implies $T_i \subseteq X_{OPT}$).

\begin{claim}[Derived from Theorem~5.4 \cite{ChitnisCHM15dsfvs}]
\label{claim:compute_Tc}
One can find in time $2^{\bigO(k)}n^{\bigO(1)}$
a collection of pairs $(I, T_c)$ where $I \subseteq [\ell]$
and $T_c \subseteq V(D)$, such that the following hold:
\begin{enumerate}
\item the number of pairs $(I, T_c)$ produced is $k^{O(1)} \log n$
\item for every pair, $|I|+|T_c| \leq (2k+1)4^{2k+1}$
\item for at least one pair $(I,T_c)$ we
have $i_0 \in I$ if $A_{i_0}\neq \emptyset$, and for every $i \in [\ell]$
such that $A_i$ is active in $X_{OPT}$ we have $T_i \subseteq T_c$
\end{enumerate}
\end{claim}

\begin{proof}
Assume that $A_{i_0} \neq \emptyset$ as otherwise the result is easier,
and let $uv \in A_{i_0}$ with $u, v \in R$. We begin by computing a
subset $U \subseteq V(D)$ such that $v \in U$ and $U \cap X_{OPT}=\emptyset$.
This can be done randomly with success probability $\Theta(1/k)$
by sampling every vertex independently with probability $1/k$, 
but the process can also be derandomized by a \emph{$(n,k,k^2)$-splitter};
see Naor et al.~\cite{NaorSS95}. In particular, in polynomial time
we can compute a family of subsets $U_i \subseteq V(D)$
such that the family contains $k^{\bigO(1)}\log n$ members
and at least one member meets the conditions for $U$.
We repeat the steps below for every member $U_i$ in the family.

From now on, let us assume that we have such a set $U$.
Create a graph $D'$ as follows. For every $v \in V(D)$,
create two vertices $v^-$, $v^+$. For every $i \in [\ell]$,
create a vertex $z_i$ and add the arcs $\{u^+z_i \mid u \in S_i\}$
and $\{z_iv^- \mid v \in T_i\}$. For every arc $uv \in E(D)$, 
add the arc $u^+v^+$. Finally, add vertices $s$ and $t$,
the arc $sy_r^+$, and the arc $v^-t$ for every $v \in V(D)$. 
Finally, for every vertex $v \in U$ give $v^-$ capacity $2k+2$
by replacing $v^-$ by a set of $2k+2$ false twins. 
Let $T_c'$ be the union of all important $(s,t)$-cuts in $D'$
of size at most $2k+1$. 
By Property~\ref{prop:enum_important_cut}, $T'_c$ can be computed in time $2^{\bigO(k)}n^{\bigO(1)}$ 
and $|T'_c| \leq (2k+1)4^{2k+1}$.
Finally we set $I=\{i \mid z_i \in T_c'\}$ and $T_c=\{v \in V(D) \mid v^- \in T_c'\}$. 
Clearly $|I|+|T_c| \leq |T'_c| \leq (2k+1)4^{2k+1}$.

We claim that $I$ contains $i_0$, and that for every $A_i$ that is 
active in $X_{OPT}$ we have $T_i \subseteq T_c$.
Indeed, define the set $X'=\{v^-, v^+ \mid v \in X_{OPT}\} \cup \{z_{i_0}\}$
and recall by assumption that $X_{OPT} \cap U = \emptyset$.
Note that $X'$ is an $(s,t)$-cut. Indeed, assume to the contrary
that there is an $(s,t)$-path $P$ in $D'-X'$. Then the last arcs
of $P$ must be $u^+z_i$, $z_iv^-$ and $v^-t$ for some $i \in [\ell]$, $uv \in A_i$. 
We may also assume that the entire prefix of $P$ before $z_i$ 
visits only $s$ and vertices $w^+$, $w \in V(D)$. 
But then that prefix proves $u \in R$; $z_i \notin X'$ implies $i \neq i_0$;
and $v^- \notin X'$ implies $v \notin X_{OPT}$. This contradicts
that only $A_{i_0}$ is strongly connected to $y_r$ in $D-X_{OPT}$. 
Also note $|X'| \leq 2k+1$. Now by Property~\ref{prop:ext_important_cut}
we can push $X'$ to an important $(s,t)$-cut $X''$ of size at most $2k+1$,
hence $X'' \subseteq T_c'$. 

We claim that $z_{i_0} \in X''$ and for every $A_i$ active in $X_{OPT}$
we have $\{v^- \mid v \in T_i\} \subseteq X''$.
For the former, by assumption $u \in R$, hence either $z_{i_0} \in X''$
or the cut has been pushed closer to $t$. But since $v \in U$
and $v$ has been given high capacity, pushing the cut past $z_{i_0}$
would contradict the size bound of $2k+1$. Hence $z_{i_0} \in X''$.
For the latter, assume that $A_i$ is active in $X_{OPT}$.
Then there is a vertex $u' \in S_i \cap R$, hence $z_i \in R$, 
and the cut cannot push past the vertices $v^-$, $v \in T_i$
since $v^-t \in E(D')$. 
%
%
\end{proof}

Now we can guess the correct pair $(I,T_c)$.
Therefore, we can guess $i_0 \in I$ (or the case that $A_{i_0} = \emptyset$) and $X_{OPT} \cap T_c$, and remove these vertices from $D$.
This multiplies the running time by at most $(2k+1)4^{2k+1} \binom{(2k+1)4^{2k+1}}{k}\log n = 2^{\bigO(k^2)}\log n$, and now we can assume that
for every $i \in [\ell]$ except $i_0$, $A_i$ is not active.
Furthermore, if $A_{i_0} \neq \emptyset$ then we add all arcs $\{y_r\} \times T_{i_0}$ to the graph.
Next claim shows how to start the construction of a solution using these assumptions.

\begin{claim}
Adding the arcs $\{y_r\} \times T_{i_0}$ does not affect the solution.
Furthermore, let $S$ be the set of vertices in $X_{OPT}$ which have an
in-neighbour reachable from $y_r$ in $D-X_{OPT}$.
There exists an important $(\{y_r\}, Y \setminus \{y_r\} \cup \bigcup_{i \neq i_0} S_i)$-cut $S'$ of size at most $|S|$
such that $(X_{OPT} \setminus S) \cup S'$ is a solution to $I$. 
\end{claim}

\begin{proof}
We first note that since $R \cap S_{i_0} \neq \emptyset$, then for every $v \in T_{i_0}$ either $v \in R$ or $v \in X_{OPT}$ 
(for example due to blocking paths from $y_r$ to some $y_i$, $i<r$).
Hence adding the arcs $\{y_r\} \times T_{i_0}$ has no effect on the solution. However,
it does simplify the important separator step below.

Now observe that $S$ is a $(\{y_r\}, Y \setminus \{y_r\} \cup \bigcup_{i \neq i_0} S_i)$-cut.
By Property~\ref{prop:ext_important_cut}, there exists an important $(\{y_r\}, Y \setminus \{y_r\} \cup \bigcup_{i \neq i_0} S_i)$-cut $S'$ with $|S'| \leq |S|$ such that
every vertex reachable from $y_r$ in $D-S$ is still
reachable from $y_r$ in $D-S'$.
We prove that $X' := (X_{OPT} \setminus S) \cup S'$ is a solution for $I$.
Clearly $|X| \leq k$, so we only need to show that $X'$ cuts all the 
closed walks intersecting several of the sets $A_1, \dots, A_\ell$
and all $(y_j,y_i)$-paths, $j>i$.

Suppose for contradiction that there exists two distinct indices $i \neq j$
and a closed walk $W$ such that $W$ intersects both $A_i$ and $A_j$.
First, $i\neq i_0$ and $j \neq i_0$: since the arc $y_rv$ is added for every $v \in T_{i_0}$,
either $v \in X_{OPT}$ or $v \in R$. Thus there is no path from $T_{i_0}$ to $S_i$ for any $i \neq i_0$
in $D-X'$ by the choice of the cut $S'$. 
Moreover, $W$ must intersect $S$, as otherwise $W$ is a closed walk in $D-X_{OPT}$,
contradicting the fact that $X_{OPT}$ is a solution.
Let $s$ be a vertex in $S \cap W$, then either $s \in S'$, and so $S'$ intersects $W$;
or $s$ is reachable from $y_r$ in $D-S'$.
But then $S_i$ is reachable from $y_r$ in $D-S'$, contradicting the fact that
$S'$ is an $(y_r, \bigcup_{i \neq i_0} S_i)$-cut.
This contradiction proves that $X'$ is a solution.
By a similar argument, $X'$ also cuts all $(y_j,y_i)$-paths for $j>i$. 
\end{proof}

Note that $(X_{OPT} \setminus S) \cup S'$ might have a non empty shadow.
This is not a problem as we will apply the shadow removal procedure at each step.

We can now give the algorithm $\mathcal{A}'$ on the instance $(D,(A_i),k,Y)$
of {\sc Arc Terminal Symmetric Directed Multiway Cut Compression}:
\begin{enumerate}
    \item reduce to the shadowless case by applying
        Subsection~\ref{subsec:reduce_shadow};
    \item compute (and guess) $(I,T_c)$ with Claim~\ref{claim:compute_Tc},
         guess $i_0 \in I \cup \{0\}$ and $X_c:=X_{OPT} \cap T_c \subseteq T_c$;
    \item let $D'=D-X_c$, and if $i_0 \neq 0$, add all arcs $\{y_r\} \times T_{i_0}$;
    \item guess an important $(\{y_r\}, Y \setminus \{y_r\} \cup \bigcup_{i \neq i_0} S_i)$-cut $S$
    of size at most $k-|X_c|$ in $D'$;
    \item if $\mathcal{A}'(D-S-X_c, (A_i), k-|S|-|X_c|, Y\setminus\{y_r\})$ returns
     a solution $X'$, return $S \cup X_c \cup X'$; otherwise proceed with the next guess
     or return ``no solution''.
\end{enumerate}

First, it is easy to see that if this algorithm returns a set $X$,
then $X$ is a solution of the input instance.
Moreover, by all the previous claims, if there exists a solution, then there
exists a sequence of guesses which will find it.
This algorithms explores a tree of depth at most $k$ with maximum degree
$2^{\bigO(k^2)}\log^3 n$, and each node is processed in time $2^{\bigO(k^2)}n^{\bigO(1)}$.
Hence the total running time is at most
\[
\left(2^{\bigO(k^2)} \log^3 n \right)^k 2^{\bigO(k^2)} n^{\bigO(1)}
= 2^{\bigO(k^3)} n^{\bigO(1)}
\]
using in particular Lemma~\ref{lem:logn_pow_k} from the appendix.
This completes the proof of Theorem~\ref{thm:multiway_is_fpt}.
\qed

\bibliographystyle{plainurl}
\bibliography{biblio}

{\Huge Appendix}
\appendix

\begin{proof}[Proof of Proposition~\ref{prop:iterative_compression}]
Let $\mathcal{A}'(D,T,k)$ be an algorithm solving the problem
$\mathcal{L}'$ in time $t(k,|Y|)n^c$.
We now solve the original problem $\mathcal{L}$ as follows.
Consider an arbitrary ordering $v_1, \dots,v_n$ of $V(D)$.
We will compute iteratively a set $X_i \subseteq \{v_1, \dots v_i\}$
of size at most $\alpha k$ which is a solution of the partial instance
$I_i$ induces by $\{v_1, \dots v_i\}$.

We start with $X_0 = \emptyset$, which is a solution of $I_0$ by assumption.
Then, if $V_i$ is already computed, we apply $\mathcal{A}'$
to $(D[\{v_1,\dots,v_{i+1}\}], T,X_i\cup\{v_{i+1}\},k)$, which returns
by assumption a solution of size at most $\alpha k$, or says that
there is no solution of size at most $k$, and in this latter case we return ''no''
directly.
This call is valid because $X_i \cup\{v_{i+1}\}$ is a solution of
$(D[\{v_1,\dots,v_{i+1}\}], T,X_i\cup\{v_{i+1}\})$ of size at most $\alpha k +1$.

This algorithm consists in $n$ calls to $\mathcal{A'}$ with the solution
to compress of size at most $\alpha k +1$. Hence its running time
is at most $t(k, \alpha k +1)n^{c+1}$.
\end{proof}

\begin{lemma}\label{lem:logn_pow_k}
If $n \geq 2^{16}$ and $p \geq 0$, then $(\log n)^p \leq n + p^{2p}$.
\end{lemma}

\begin{proof}
If $p \geq \sqrt{\log n}$ then $n \leq 2^{p^2}$ and $(\log n)^p \leq p^{2p}$.

Otherwise, $p<\sqrt{\log n}$.
First, we show the following property:
\begin{equation*}
    n \geq 2^{16} \Rightarrow \sqrt{\log n} \leq \frac{\log n}{\log \log n}
\end{equation*}
To prove that, note that this property is equivalent to $2\log N \leq N$
with $N = \sqrt{\log n}$.
Then $N \geq 4$ is a sufficient condition, and $n \geq 2^{16}$ too.
Now we apply this result and we get 
$p \leq \sqrt{\log n} \leq \frac{\log n}{\log \log n}$.
It follows that $(\log n)^p \leq n$.
\end{proof}

\end{document}